\documentclass[a4paper,UKenglish,cleveref, autoref, thm-restate]{lipics-v2021}

\hideLIPIcs  

\title{ARRIVAL: Recursive Framework \texorpdfstring{\&}{and} \texorpdfstring{$\ell_1$}{l1}-Contraction} 

\author{Sebastian Haslebacher}{ETH Zurich, Switzerland}{sebastian.haslebacher@inf.ethz.ch}{https://orcid.org/0000-0003-3988-3325}{}

\authorrunning{S. Haslebacher} 

\Copyright{Sebastian Haslebacher} 



\relatedversion{}
\relatedversiondetails{Full Version}{https://doi.org/10.48550/arXiv.2502.06477}

\acknowledgements{I want to thank Bernd Gärtner and Simon Weber for valuable discussions and feedback, and anonymous reviewers for their comments and suggestions. }

\nolinenumbers 

\EventEditors{Keren Censor-Hillel, Fabrizio Grandoni, Joel Ouaknine, and Gabriele Puppis}
\EventNoEds{4}
\EventLongTitle{52nd International Colloquium on Automata, Languages, and Programming (ICALP 2025)}
\EventShortTitle{ICALP 2025}
\EventAcronym{ICALP}
\EventYear{2025}
\EventDate{July 8--11, 2025}
\EventLocation{Aarhus, Denmark}
\EventLogo{}
\SeriesVolume{334}
\ArticleNo{62}

\usepackage{mathtools}
\usepackage{amsmath, amssymb, amsfonts, mathrsfs}
\usepackage[linesnumbered, ruled, vlined]{algorithm2e}

\newcommand{\NP}{\mathsf{NP}}

\newcommand{\CoNP}{\mathsf{CoNP}}
\newcommand{\UP}{\mathsf{UP}}
\newcommand{\CoUP}{\mathsf{CoUP}}
\renewcommand{\P}{\mathsf{P}}

\newcommand{\UEOPL}{\mathsf{UEOPL}}

\newcommand{\TFNP}{\mathsf{TFNP}}

\newcommand{\N}{\mathbb{N}}

\newcommand{\R}{\mathbb{R}}

\newcommand{\bigO}{\mathcal{O}}
\DeclareMathOperator{\poly}{poly}

\DeclarePairedDelimiter\norm{\lVert}{\rVert}

\let\emptyset\varnothing

\renewcommand{\epsilon}{\ensuremath\varepsilon}

\renewcommand{\phi}{\ensuremath{\varphi}}

\newcommand{\problem}[1]{\textrm{#1}}
\newcommand{\arrival}{\problem{ARRIVAL}}
\newcommand{\ssg}{\problem{SSG}}
\newcommand{\Tarski}{\problem{Tarski}}

\newcommand{\garrival}{\problem{G-ARRIVAL}}

\begin{document}

\maketitle

\begin{abstract}
\arrival\ is the problem of deciding which out of two possible destinations will be reached first by a token that moves deterministically along the edges of a directed graph, according to so-called switching rules. It is known to lie in $\NP \cap \CoNP$, but not known to lie in $\P$. The state-of-the-art algorithm due to Gärtner et al.\@~(ICALP `21) runs in time $2^{\bigO(\sqrt{n} \log n)}$ on an $n$-vertex graph.

We prove that \arrival\ can be solved in time $2^{\bigO(k \log^2 n)}$ on $n$-vertex graphs of treewidth $k$. Our algorithm is derived by adapting a simple recursive algorithm for a generalization of \arrival\ called \garrival. This simple recursive algorithm acts as a framework from which we can also rederive the subexponential upper bound of Gärtner et al.

Our second result is a reduction from \garrival\ to the problem of finding an approximate fixed point of an $\ell_1$-contracting function $f : [0, 1]^n \rightarrow [0, 1]^n$. Finding such fixed points is a well-studied problem in the case of the $\ell_2$-metric and the $\ell_\infty$-metric, but little is known about the $\ell_1$-case. 

Both of our results highlight parallels between \arrival\ and the Simple Stochastic Games (\ssg) problem. Concretely, Chatterjee et al.\@~(SODA `23) gave an algorithm for \ssg\ parameterized by treewidth that achieves a similar bound as we do for \arrival, and \ssg\ is known to reduce to $\ell_\infty$-contraction.

\end{abstract}

\section{Introduction}
\label{sec:introduction}

\arrival\ is a computational problem first introduced by Dohrau et al.\@~\cite{dohrauARRIVALZeroPlayerGraph2017}. It can be described as a deterministic process (or zero-player game) on a directed graph with a designated origin $o$ and two designated destinations $d$ and $\overline{d}$. Every vertex in \arrival\ has out-degree two, and exactly one outgoing edge at every vertex is marked (we also call it the \emph{even} edge). We additionally assume that both destinations are reachable from every vertex in the graph. A token is placed on $o$ and moved along the edges of the graph according to the following rule: At every vertex, the token continues along the outgoing edge that was used least so far. In case of a tie, the token uses the even edge. This effectively means that the token will alternate between the two outgoing edges at every vertex, starting with the even edge. The task is to decide which of the two destinations $d$ or $\overline{d}$ will be visited first by the token.

Dohrau et al.\@~\cite{dohrauARRIVALZeroPlayerGraph2017} proved that \arrival\ is contained in $\NP \cap \CoNP$. Naturally, they then asked whether it is also in $\P$. This open problem has received some attention in recent years and the best algorithm to date, due to Gärtner et al.\@~\cite{gartnerSubexponentialAlgorithmARRIVAL2021}, runs in time $2^{\bigO(\sqrt{n} \log n)}$ on a graph with $n$ vertices.

We present two new results for \arrival: Our first result is an algorithm for \arrival\ that runs in time $2^{\bigO \left(k \log^2 n \right)}$ on graphs with $n$ vertices and treewidth $k$. Note that this bound is quasi-polynomial for graphs with bounded treewidth. Our algorithm is obtained by adapting a simple recursive algorithm for \garrival, a generalization of \arrival\ that allows arbitrarily many origins, destinations, and tokens: The recursive algorithm solves an instance with $\ell$ destinations and origins by making recursive calls on instances with $\ell + 1$ destinations and origins. In other words, in each recursive call, a new vertex is made into a destination and origin vertex. By choosing this pivot vertex carefully, we can exploit the underlying graph structure.

It turns out that this simple recursive algorithm for \garrival\ can also be seen as a framework for other algorithms for \garrival. Concretely, we explain how the state-of-the-art upper bound $2^{\bigO(\sqrt{n}\log n)}$ due to Gärtner et al.\@~\cite{gartnerSubexponentialAlgorithmARRIVAL2021} can be derived in this framework as well.

Our second result is a reduction from \garrival\ to the problem of finding an approximate fixed point of a $\ell_1$-contracting function $f : [0 ,1]^n \rightarrow [0 ,1]^n$. Concretely, we say that a function $f : [0 ,1]^n \rightarrow [0 ,1]^n$ is contracting with parameter $\lambda \in [0, 1)$ if we have $\norm{f(x) - f(y)}_1 \leq \lambda \norm{x - y}_1$ for all $x, y \in [0, 1]^n$. Such a function is guaranteed to have a unique fixed point by Banach's fixed point theorem~\cite{banach1922operations}. An $\epsilon$-approximate fixed point $x \in [0, 1]^n$ has to satisfy $\norm{f(x) - x}_1 \leq \epsilon$. Given our reduction, any algorithm that can find an $\epsilon$-approximate fixed point of $f$ in time $\poly \left(\log \frac{1}{\epsilon}, \log \frac{1}{1 - \lambda}, n \right)$ would imply a polynomial-time algorithm for \garrival.

Finding approximate fixed points is a well-studied problem in the case of the $\ell_2$-metric and the $\ell_\infty$-metric (see e.g.\@~\cite{sikorskiComputationalComplexityFixed2009, chenComputingFixedPoint2024}). In particular, efficient algorithms in the case of the $\ell_2$-metric have been known since 1993~\cite{sikorskiEllipsoidAlgorithmComputation1993}, and Chen et al.\@~\cite{chenComputingFixedPoint2024} only recently gave the first polynomial query upper bound for the $\ell_\infty$-metric. Even more recently, a generalization of the result by Chen et al.\ to $\ell_p$-contractions for every $p \in [1, \infty]$ was announced~\cite{haslebacherQueryEfficientFixpoints$ell_p$Contractions2025}. Concretely, this means that an approximate fixed point of an $\ell_1$-contraction can be found with polynomially many queries to the contraction map. 

Unfortunately, the algorithm for $\ell_1$-contractions in~\cite{haslebacherQueryEfficientFixpoints$ell_p$Contractions2025} is only query-efficient (and not time-efficient). Thus, our reduction currently does not imply better algorithms for \garrival. We still think that our reduction is interesting: To the best of our knowledge, it provides the first concrete application for the $\ell_1$-contraction problem. Moreover, the polynomial query upper bound for $\ell_1$-contractions~\cite{haslebacherQueryEfficientFixpoints$ell_p$Contractions2025} suggests that better (maybe even time-efficient) algorithms could be obtained for $\ell_1$-contraction and hence \garrival\ in the future. For this, it might also be interesting to observe that the $\ell_1$-contraction map obtained through our reduction from \garrival\ has the additional property of being monotone with respect to the coordinate-wise partial order. Given recent work of Batziou et al.\@~\cite{batziouMonotoneContractions2024} on monotone $\ell_\infty$-contractions, it seems plausible that $\ell_1$-contraction and monotonicity could maybe be exploited simultaneously as well.

\subsection{Related Work}

Since \arrival\ is contained in $\NP \cap \CoNP$, it also naturally fits into the complexity class $\TFNP$, which contains total problems with efficiently verifiable solutions. In fact, after a series of results for containment in subclasses of $\TFNP$~\cite{karthikc.s.DidTrainReach2017, gartnerARRIVALNextStop2018}, we now know that \arrival\ is contained in $\UEOPL$~\cite{fearnleyUniqueEndPotential2020}.

Going into a slightly different direction, Gärtner et al.\@~\cite{gartnerARRIVALNextStop2018} proved that \arrival\ is also contained in $\UP$ and $\CoUP$, the analogues of $\NP$ and $\CoNP$ with unique solutions. 
In fact, it would not be hard to rederive this using our reduction to $\ell_1$-contraction: The idea is that the fixed point acts as an efficient certificate for both YES- and NO-instances and it must be unique due to the contraction property.

It is also known that \arrival\ can be solved in polynomial time on some restricted graph classes. For example, a result due to Priezzhev et al.\@~\cite{priezzhevEulerianWalkersModel1996} implies that \arrival\ can be solved by simulation in polynomial time on Eulerian graphs. Other results include polynomial-time algorithms on tree-like multigraphs~\cite{augerPolynomialTimeAlgorithm2022} and path-like multigraphs with many tokens~\cite{augerGeneralizedARRIVALProblem2023}. More recently, algorithms running in quasi-polynomial-time were given for the case of tree-like multigraphs with many tokens~\cite{ghorbaniQuasiPolynomialTimeAlgorithm2025}. Our algorithm effectively generalizes this result to all graphs of bounded treewidth.

Finally, further variants of \arrival\ have been studied in the past, including a stochastic variant~\cite{websterStochasticArrivalProblem2022}, a recursive variant~\cite{websterRecursiveArrivalProblem2023}, as well as variants with one or two players~\cite{fearnleyReachabilitySwitchingGames2021}.

\paragraph*{Comparison to \ssg\ }

As mentioned before, both our results show parallels between \arrival\ and Simple Stochastic Games (\ssg). We will briefly discuss the connections between the two problems. 
Similarly to \arrival, \ssg\ is contained in $\NP \cap \CoNP$~\cite{condonComplexityStochasticGames1992} and even $\UP \cap \CoUP$~\cite{chatterjeeReductionParityGames2011a}, but no polynomial-time algorithm is known. The state-of-the-art algorithm due to Ludwig~\cite{ludwigSubexponentialRandomizedAlgorithm1995} runs in randomized subexponential time $2^{\bigO(\sqrt{n}\log n)}$.

Dohrau et al.\@~\cite{dohrauARRIVALZeroPlayerGraph2017} already wondered about similarities between \arrival\ and \ssg\ when they first introduced \arrival. Since then, both problems were shown to reduce to the problem of finding a Tarski fixed point~\cite{etessamiTarskiTheoremSupermodular2020, gartnerSubexponentialAlgorithmARRIVAL2021}, and to be contained in the complexity class $\UEOPL$~\cite{fearnleyUniqueEndPotential2020}. Both our results for \arrival\ further extend this list of similarities: Our upper bound of $2^{\bigO(k \log^2 n)}$ for \arrival\ on $n$-vertex graphs of treewidth $k$ is comparable to a similar bound for \ssg\ due to Chatterjee et al.\@~\cite{chatterjeeFasterAlgorithmTurnbased2023}. Moreover, \ssg\ reduces to finding an approximate fixed point of an $\ell_\infty$-contracting function~\cite{condonComplexityStochasticGames1992}, which is analogous to our reduction from \arrival\ to $\ell_1$-contraction. 

Both \arrival\ and \ssg\ also admit polynomial-time algorithms on graphs with a bounded feedback vertex set~\cite{gartnerSubexponentialAlgorithmARRIVAL2021, augerFindingOptimalStrategies2014}. In fact, in the case of \arrival, we will discuss this in more detail in Section~\ref{ssec:subexponential_algo}.

\subsection{Outline}

As explained above, our results are actually obtained for a generalization of \arrival\ called \garrival. Thus, we will start Section~\ref{sec:preliminaries} by formally introducing \garrival. We also use Section~\ref{sec:preliminaries} to recall further terminology and notation from the literature that will be useful for our arguments. 

Note that instead of formally introducing the notion of treewidth and tree decompositions, we will directly work with so-called balanced separators instead. The reason for this choice of exposition is that our parameterized algorithm actually exploits the existence of small balanced separators (and not tree decompositions themselves), which are guaranteed to exist in graphs of small treewidth (see Section~\ref{ssec:treewidth_separators} for more details). 

We describe our parameterized algorithm in Section~\ref{sec:algorithm}. We start the exposition with our simple recursive algorithm for \garrival\ (Section~\ref{ssec:recursive_algorithm}), which provides a framework from which we will derive the parameterized algorithm in Section~\ref{ssec:exploiting_balanced_separators}. In Section~\ref{ssec:subexponential_algo}, we additionally explain how the subexponential upper bound due to Gärtner et al.\@~\cite{gartnerSubexponentialAlgorithmARRIVAL2021} and polynomial-time upper bounds on graphs with a bounded feedback vertex set~\cite{gartnerSubexponentialAlgorithmARRIVAL2021} can be rederived from our framework.

Finally, Section~\ref{sec:solving_arrival_via_contraction} contains our reduction to the problem of finding an approximate fixed point of a $\ell_1$-contraction map. Crucially, given an instance of \garrival, we define a function $f : \R^n_{\geq 0} \rightarrow \R^n_{\geq 0}$ that is contracting and thus has a unique fixed point. We then prove that a reasonably good approximation of this fixed point will give away the solution to the \garrival-instance. The function can be restricted to a compact subset of $\R^n_{\geq 0}$ and scaled to fit into $[0, 1]^n$, if desired.

\section{Preliminaries}
\label{sec:preliminaries}

We start by recalling \garrival, which was first formally defined by Hoang~\cite{hoangTwoCombinatorialReconfiguration2022}. Note that our formulation slightly deviates from the one by Hoang, but is easily seen to be equivalent.

A switch graph is a directed graph $G = (V, E, s_0, s_1)$ with $s_0, s_1 : V \rightarrow V$ and $E = \{ (v, s_0(v)) \mid v \in V \} \cup \{ (v, s_1(v)) \mid v \in V \}$.
We consider $E$ to be a multiset, and two edges $(v, s_0(v)), (v, s_1(v)) \in E$ to be distinct objects even if we have $s_0(v) = s_1(v)$. Given a switch graph, we call $s_0(v)$ and $s_1(v)$ the even and odd successor of $v \in V$, respectively. Similarly, we refer to edges induced by $s_0$ as even edges, and to edges induced by $s_1$ as odd edges. 

In \garrival, many tokens traverse the directed graph simultaneously, starting and ending at special vertices that we call terminals or terminal vertices. Concretely, a \garrival-instance consists of a switch graph as well as a non-empty subset $T \subseteq V$ of terminals. For each terminal $v \in T$, we also get a natural number $t^+_v$ of tokens starting at $v$. We always assume that at least one terminal is reachable from each non-terminal vertex $v \in V \setminus T$ (otherwise, tokens could loop indefinitely without ever reaching a terminal). We will also always denote the total number of tokens by $t^+ \coloneqq \sum_{v \in T} t^+_v$ and assume $t^+ \geq 1$.

\begin{figure}[ht]
    \centering
    \includegraphics[width=0.6\linewidth]{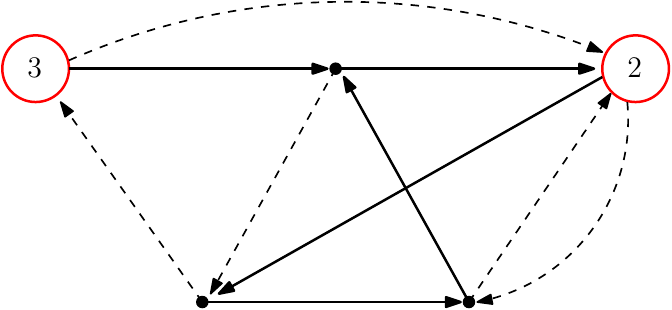}
    \caption{An instance of \garrival. Even edges are bold while odd edges are dashed. Terminals are marked in red. Three tokens start at the left terminal, and two tokens start at the right terminal. }
    \label{fig:g-arrival}
\end{figure}

Now consider the following non-deterministic procedure. Initially, for every terminal $v \in T$, move $\lceil \frac{t^+_v}{2} \rceil$ tokens from $v$ to $s_0(v)$, and $\lfloor \frac{t^+_v}{2} \rfloor$ tokens from $v$ to $s_1(v)$ (we do this for every terminal simultaneously). Then, while there exists a token on a non-terminal $v \in V \setminus T$, non-deterministically choose one such token and move it along the out-edge of $v$ that was used fewer times so far. In case of a tie, the token must use the even out-edge $(v, s_0(v)) \in E$. In other words, the even and odd out-edges at $v$ will be used in an alternating fashion, starting with the even out-edge. This procedure stops once all tokens have reached a terminal. The goal of \garrival\ is to predict the number of tokens $t^-_v$ arriving at each terminal $v \in T$. 

In order for \garrival\ to be well-defined, we need to make sure that the above procedure terminates and that it always produces the same values $t^-_v$ for all $v \in T$ (independently of the non-deterministic choices). Indeed, both of these properties hold, and this follows by generalizing the arguments of Dohrau et al.\@~\cite{dohrauARRIVALZeroPlayerGraph2017} to work with multiple tokens, multiple destinations, and multiple origins, as explained by Gärtner et al.\@~\cite{gartnerSubexponentialAlgorithmARRIVAL2021} and Hoang~\cite{hoangTwoCombinatorialReconfiguration2022}. To explain this, we first need to recall the concept of switchings flows from the literature. 

Given a switch graph $G = (V, E, s_0, s_1)$ with terminals $T \subseteq V$ and starting tokens $(t^+_v)_{v \in T}$, a function $x : E \rightarrow \R_{\geq 0}$ satisfying the three constraints
\begin{align*}
    x(v, s_0(v)) - x(v, s_1(v)) &\in \{0, 1\} &(\forall v \in V)\\
    \underbrace{\sum_{u : (v, u) \in E} x(v, u)}_{ =: x^+(v) } - \underbrace{\sum_{u : (u, v) \in E} x(u, v)}_{:= x^-(v)} &= 0 &(\forall v \in V \setminus T) \\
    \underbrace{\sum_{u : (v, u) \in E} x(v, u)}_{ =: x^+(v) } &= t^+_v &(\forall v \in T) \\
\end{align*}
is called a switching flow. We will refer to the first set of constraints above as switching behavior, and to the second set of constraints as flow conservation. 

Dohrau et al.\@~\cite{dohrauARRIVALZeroPlayerGraph2017} proved the following theorem in the case of \arrival, but we directly state its generalization for \garrival\ (see also~\cite{gartnerSubexponentialAlgorithmARRIVAL2021, hoangTwoCombinatorialReconfiguration2022}).

\begin{theorem}[Integral Switching Flows are Certificates~\cite{dohrauARRIVALZeroPlayerGraph2017}]
\label{theorem:switching_flows_are_certificates}
    Given a switch graph $G = (V, E, s_0, s_1)$ with terminals $\emptyset \neq T \subseteq V$ and starting tokens $(t^+_v)_{v \in T}$ with $t^+ \geq 1$, the number of tokens $t^-_v$ arriving at terminal $v$ is well-defined and any integral switching flow $x : E \rightarrow \N_0$ satisfies $x^-(v) = t^-_v$, for all $v \in T$.
\end{theorem}

Concretely, Theorem~\ref{theorem:switching_flows_are_certificates} states that \garrival\ is well-defined and that it can be solved by finding any integral switching flow. Observe that by simulating the non-deterministic procedure outlined before and recording the number of times that each edge is traversed by a token, one can obtain a special integral switching flow that we call the run profile. While the run profile is unique (i.e.\ it does not depend on non-deterministic choices)~\cite{gartnerSubexponentialAlgorithmARRIVAL2021}, Dohrau et al.\@~\cite{dohrauARRIVALZeroPlayerGraph2017} already observed that in general, integral switching flows are not unique. Concretely, there may be integral switching flows other than the run profile, but Theorem~\ref{theorem:switching_flows_are_certificates} tells us that they still predict the values $(t^-_v)_{v \in T}$ correctly. 

\begin{figure}[ht]
    \centering
    \includegraphics[width=0.6\linewidth]{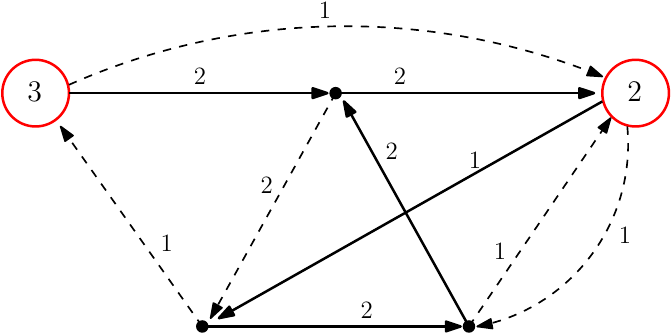}
    \caption{The number on the edges indicate an integral switching flow. Note that this is not the run profile, but it still certifies (by Theorem~\ref{theorem:switching_flows_are_certificates}) that four tokens arrive at the right terminal while only one token arrives at the left terminal. To get the run profile, one would have to decrease the flow on the edges of the directed triangle formed by the three non-terminal vertices by one each.}
    \label{fig:switching_flow}
\end{figure}

Finally, we will need an upper bound on the total flow in any integral switching flow. Similar upper bounds were used in previous work as well~(see e.g.\ \cite{dohrauARRIVALZeroPlayerGraph2017, gartnerARRIVALNextStop2018, gartnerSubexponentialAlgorithmARRIVAL2021, hoangTwoCombinatorialReconfiguration2022}). We sketch a short proof.

\begin{lemma}[Upper Bound on Integral Switching Flow~\cite{dohrauARRIVALZeroPlayerGraph2017}]
\label{lemma:upper_bound_flow}
    Let $x : E \rightarrow \N_0$ be an arbitrary integral switching flow for the \garrival-instance on a switch graph $G = (V, E, s_0, s_1)$ with terminals $\emptyset \neq T \subseteq V$ and starting tokens $(t^+_v)_{v \in T}$ with $t^+ \geq 1$. Then we must have $t^+ = \sum_{v \in T} t^+_v = \sum_{v \in T} t^-_v = \sum_{v \in T} x^-(v)$ and $x(e) < 2^{|V|}t^+$ for all $e \in E$.
\end{lemma}
\begin{proof}
    The equation $\sum_{v \in T} t^+_v = \sum_{v \in T} t^-_v = \sum_{v \in T} x^-(v)$ follows from flow conservation of switching flows and Theorem~\ref{theorem:switching_flows_are_certificates}. We will now prove the upper bound on the flow values. 
    
    Let $e = (u, v) \in E$ be arbitrary. Observe that there must exist a simple path $P = (v_0, v_1, v_2, \dots, v_k)$ with $v_0 = v$ of length $0 \leq k < |V|$ from $v$ to some $v_k \in T$ (recall that we assume that at least one terminal is reachable from every non-terminal in the graph). Observe that by the switching behavior of switching flows, $x(e) \geq 2^{|V|}t^+$ would imply $x^-(v_i) \geq 2^{|V| - i}t^+$ for all $i \in [k]$. In particular, we would have $t^-_{v_k} = x^-(v_k) \geq 2 t^+$, contradicting the equation $t^+ = \sum_{v \in T} t^-_v$ above.
\end{proof}

\subsection{Treewidth and Balanced Separators}
\label{ssec:treewidth_separators}

Treewidth is a well-established graph parameter that plays an important role in parameterized algorithms, and it is intimately related to the notion of balanced separators (see e.g.\@~\cite[Chapter~7]{cyganParameterizedAlgorithms2015}). As is the case with many applications on graphs of bounded treewidth, our algorithm actually exploits the existence of balanced separators and can be formulated without computations of tree decompositions. Hence, we will refrain from formally introducing tree decompositions and instead focus on balanced separators. 

Given an undirected graph $G = (V, E)$ and a subset $S \subseteq V$ of its vertices, we use $G - S$ to denote the graph resulting from deleting the vertices in $S$ and their incident edges from $G$. We call $S$ a balanced separator if each connected component in $G - S$ contains at most $\frac{1}{2}|V|$ vertices. Note that this does not necessarily imply that $G - S$ must have more than one connected component and it could even be an empty graph (despite what the term separator may suggest): Concretely, any set $S$ of size at least $\frac{1}{2}|V|$ is a balanced separator, and thus there always exists a balanced separator. Using a brute-force approach, we can find a smallest balanced separator $S$ in $G$ in time $|V|^{\bigO(|S|)}$.

The following connection between treewidth and balanced separators is crucial for us.

\begin{lemma}[{Balanced Separators and Treewidth~\cite[Lemma 7.19]{cyganParameterizedAlgorithms2015}}]
\label{lemma:existence_balanced_separator}
    If $G$ has treewidth at most~$k$, then for every $S \subseteq V$, the subgraph $G - S$ has a balanced separator of size at most $k + 1$.
\end{lemma}
This lemma allows us to use the treewidth of $G$ to infer the existence of small balanced separators in all induced subgraphs of $G$.

All of the previous concepts are defined on undirected graphs. Since we will be working exclusively with directed graphs, we will adopt the following convention: When used on a directed graph, the terms treewidth and balanced separators are to be interpreted with respect to the underlying simple undirected graph. 

\subsection{Non-Expansive, Contracting, and Monotone Functions}
\label{ssec:contraction_maps}

We will be interested in the Manhattan distance, which is induced by the $\ell_1$-norm. Concretely, we use $\norm{x} := \sum_{i = 1}^n |x_i|$ to denote the $\ell_1$-norm of a vector $x \in \R^n$. The Manhattan distance of $x, y \in \R^n$ is then given by $\norm{x - y}$. 

We are mainly concerned with the non-negative orthant $\R^n_{\geq 0} \subseteq \R^n$. A function $f : \R^n_{\geq 0} \rightarrow \R^n_{\geq 0}$ is called a $\lambda$-contraction (or is $\lambda$-contracting) for some $\lambda \in [0, 1)$ if and only if $\norm{f(x) - f(y)} \leq \lambda \norm{x - y}$ for all $x, y \in \R^n_{\geq 0}$. Banach's fixed point theorem~\cite{banach1922operations} implies that such a contracting function admits a unique fixed point. If $f$ only satisfies the weaker property $\norm{f(x) - f(y)} \leq \norm{x - y}$ for all $x, y \in \R^n_{\geq 0}$, we call it non-expansive instead.

In Section~\ref{sec:solving_arrival_via_contraction}, we reduce \garrival\ to the following computational problem: Given access to a $\lambda$-contraction $f : \R^n_{\geq 0} \rightarrow \R^n_{\geq 0}$, find an $\epsilon$-approximate fixed point of $f$, i.e.\ a point $x \in \R^n_{\geq 0}$ such that $\norm{f(x) - x} \leq \epsilon$. We will also point out how the domain of our function $f$ can be restricted to $[0, 1]^n$, if desired.

Another property that we need in some of our proofs is monotonicity with respect to the coordinate-wise partial order. Concretely, we call a function $f : X \subseteq \R^n_{\geq 0} \rightarrow \R^n_{\geq 0}$ monotone if and only if $x \leq y$ implies $f(x) \leq f(y)$ for all $x, y \in X$, where $\leq$ denotes coordinate-wise comparison.

\begin{lemma}[Monotonicity in \garrival\@~\cite{gartnerSubexponentialAlgorithmARRIVAL2021}]
\label{lemma:monotonicity_and_global_conservation}
    Consider a switch graph $G = (V, E, s_0, s_1)$ with terminals $\emptyset \neq T \subseteq V$. Consider the function $f : \N^{|T|}_0 \rightarrow \N^{|T|}_0$ that maps the vector $(t^+_v)_{v \in T}$ of starting tokens to the vector $(t^-_v)_{v \in T}$ of tokens ending at terminals. The function $f$ is monotone with respect to the coordinate-wise partial order.
\end{lemma}

\section{A Family of Recursive Algorithms}
\label{sec:algorithm}

The main goal of this section is to give a parameterized algorithm for \garrival\ that runs in time $2^{\bigO \left(k \log n \log (n + \log t^+) \right)}$ on graphs with $n$ vertices, treewidth $k$, and a total of $t^+$ starting tokens. Note that for the sake of simplicity, our algorithm recurses by doing a binary search over one vertex at a time. One could instead use algorithms for the so-called \Tarski-problem (see e.g.\@~\cite{etessamiTarskiTheoremSupermodular2020}) to recurse by searching over multiple vertices (e.g.\ all vertices of a balanced separator) at a time. However, with the current best algorithms for \Tarski, this would not yield any asymptotic improvements.

\subsection{A Simple Recursive Algorithm }
\label{ssec:recursive_algorithm}

We start by explaining a simple recursive algorithm for \garrival\ that is inspired by the approach of Gärtner et al.\@~\cite{gartnerSubexponentialAlgorithmARRIVAL2021}: The algorithm chooses an arbitrary non-terminal $p \in V \setminus T$ that we call the pivot, and makes a guess $a \in \N_0$ for the outflow of $p$ in an integral switching flow. In order to verify the guess, $p$ is converted to a terminal and $t^+_p \coloneqq a$ tokens are assigned to start at $p$. In this way, we obtain again an instance of \garrival\ with one more terminal. After solving this subinstance, we can check our guess by looking at the number of tokens $t^-_p$ that arrive at $p$ in the subinstance. If we find that $t^-_p = a = t^+_p$, the integral switching flow obtained for the subinstance is also an integral switching flow for the original instance where $p$ is not a terminal. Otherwise, we have $t^-_p < a$ or $t^-_p > a$ and we use this information to adjust our guess in a binary search fashion. We make this precise in Algorithm~\ref{algo:recursive_algorithm} and use the remainder of this section to prove that Algorithm~\ref{algo:recursive_algorithm} is correct. 

\begin{algorithm}
\DontPrintSemicolon
\caption{}
\label{algo:recursive_algorithm}
\SetKwFunction{findSF}{Find-Switching-Flow-1}
\Indm\findSF{$G = (V, E, s_0, s_1), T \subseteq V, (t^+_v)_{v \in T}$}\\
\Indp
  \If{$T = V$ \tcp*{Base Case}}{
        $x(v, s_0) \gets \lceil \frac{t^+_v}{2} \rceil$ for all $v \in V = T$\\
        $x(v, s_1) \gets \lfloor \frac{t^+_v}{2} \rfloor$ for all $v \in V = T$\\
        \KwRet{$x$}
    }
    choose arbitrary $p \in V \setminus T$ \tcp*{Binary Search Case} 
    $T' \gets T \cup \{p\}$\\
    $\ell \gets 0$\\
    $r \gets 2^{|V|} t^+$\\
    \While{$\ell < r$}{
        $t^+_p \gets \lceil \frac{\ell + r}{2} \rceil$\\
        $x \gets$ \findSF{$G, T', (t^+_v)_{v \in T'}$} \\
        \If{$x^-(p) < t^+_p$}{$r \gets t^+_p - 1$}
        \If{$x^-(p) = t^+_p$}{\KwRet{$x$}}
        \If{$x^-(p) > t^+_p$}{$\ell \gets t^+_p + 1$}
    }
\end{algorithm}

\begin{lemma}[Analysis of Algorithm~\ref{algo:recursive_algorithm}]
\label{lemma:correctness_algo_1}
    Given an arbitrary \garrival-instance consisting of $G = (V, E, s_0, s_1)$ with terminals $\emptyset \neq T \subseteq V$ and starting tokens $(t^+_v)_{v \in T}$ with $t^+ \geq 1$, Algorithm~\ref{algo:recursive_algorithm} correctly returns an integral switching flow in time $ 2^{\bigO \left(|V \setminus T| \log ( |V| + \log t^+ ) \right)}$.
\end{lemma}
\begin{proof}
    We start by proving correctness by induction over the size of $V \setminus T$. As a base case, observe that for $T = V$, the algorithm clearly terminates and produces an integral switching flow $x$. Thus, assume now $T \neq V$ and let $p \in V \setminus T$ be the pivot that is chosen by the algorithm. By the induction hypothesis, we can assume that all recursive calls correctly return an integral switching flow. Consider now the function $f : \{0, 1, \dots, 2^{|V| } t^+ \} \rightarrow \N_0$ that maps the guessed number of tokens $t^+_p$ to the value $x^-(p)$ returned by the recursive call. We claim that $f$ is monotone and maps $\{0, 1, \dots, 2^{|V| } t^+ \}$ to itself. Indeed, monotonicity follows directly from Lemma~\ref{lemma:monotonicity_and_global_conservation}. Moreover, we have $f(0) \geq 0$ since by definition, we must have $x^-(p) \geq 0$. For the upper bound, we recycle the argument from Lemma~\ref{lemma:upper_bound_flow}: There must exist a path $P = (p, v_1, \dots, v_k)$ of length $k < |V|$ from $p$ to some $v_k \in T$. Thus, starting $t^+_p = 2^{|V|} t^+$ tokens at $p$ would imply $x^-(v_k) \geq t^+$ by switching behavior, and hence
    \[
        x^-(p) = t^+ + t^+_p - \sum_{v \in T} x^-(v) \leq t^+ + t^+_p - x^-(v_k) \leq t^+_p
    \] 
    using Lemma~\ref{lemma:upper_bound_flow}. We conclude that binary search will successfully find a correct guess for $t^+_p$ in the given set $\{0, 1, \dots, 2^{|V|} t^+\}$.

    Having proved correctness, we move on to the bound on the overall runtime. Observe that the recursion depth of the algorithm is at most $|V \setminus T|$. Thus, the total number of tokens starting at terminals in any of the recursive call is always bounded from above by $N \coloneqq (1 + 2^{|V|})^{|V \setminus T|} t^+$. Let now $\mathcal{T}(\ell)$ denote an upper bound on the runtime of the algorithm on subinstances with $\ell$ terminals. We get that  
    \begin{align*}
        \mathcal{T}(|T|) &\leq (c \log N) \mathcal{T}(|T| + 1) \\
        &\leq (c \log N)^2 \mathcal{T}(|T| + 2) \\
        &\quad \vdots \\
        &\leq (c \log N)^{|V \setminus T|} \mathcal{T}(|V|) 
    \end{align*}
    for some constant $c$. Using the definition of $N$ and $\mathcal{T}(|V|) \leq \bigO( |V| \log N )$, this yields an overall runtime of $2^{\bigO \left(|V \setminus T| \log ( |V| + \log t^+ ) \right)}$, as desired. 
\end{proof}

\subsection{Subexponential Upper Bound}
\label{ssec:subexponential_algo}

Picking the pivot $p \in V \setminus T$ in Algorithm~\ref{algo:recursive_algorithm} arbitrarily seems quite naive. In this section, we explain how applying the ideas of Gärtner et al.\@~\cite{gartnerSubexponentialAlgorithmARRIVAL2021} yields a subexponential upper bound.

\begin{lemma}[Algorithm with Diameter-Like Bound~\cite{gartnerSubexponentialAlgorithmARRIVAL2021}]
\label{lemma:multi-run_diamter_bound}
    Consider an arbitrary \garrival-instance consisting of $G = (V, E, s_0, s_1)$ with non-empty $T \subseteq V$ and starting tokens $(t^+_v)_{v \in T}$ with $t^+ \geq 1$. Let $\ell \coloneqq \max_{v \in V \setminus T} \text{dist}(v, T)$, where $\text{dist}(v, T)$ denotes the shortest path distance from $v$ to any vertex in $T$. There is an algorithm that solves \garrival\ in time $2^\ell \poly(|V|, \log t^+)$. 
\end{lemma}

\begin{lemma}[Decomposition Lemma~\cite{gartnerSubexponentialAlgorithmARRIVAL2021}]
\label{lemma:decomposition_lemma}
    Let $G = (V, E, s_0, s_1)$ be an arbitrary switch graph with a non-empty set $T \subseteq V$ of terminals. There is an algorithm that finds a set $S \subseteq V \setminus T$ of size $\bigO(\sqrt{|V|})$ satisfying $\max_{v \in V \setminus (T \cup S)} \text{dist}(v, T \cup S) \leq \bigO(\sqrt{|V|} \log |V|)$ in time $\bigO(|V|)$.
\end{lemma}

Given those two observations, it is now not hard to adapt Algorithm~\ref{algo:recursive_algorithm} to run in subexponential-time: We first precompute the set $S$ from Lemma~\ref{lemma:decomposition_lemma} in linear time. Then, we run Algorithm~\ref{algo:recursive_algorithm} with two changes: In the recursive step, we make sure to always pick a pivot $p$ from $S$ instead of all of $V \setminus T$. Further, we add a new base case that applies the algorithm from Lemma~\ref{lemma:multi-run_diamter_bound} as soon as the parameter $\ell$ from Lemma~\ref{lemma:multi-run_diamter_bound} has shrunk to $\bigO(\sqrt{|V|} \log |V|)$, which is guaranteed to happen at the latest once all of the vertices in $S$ have been turned into terminals. With these two changes, the recursion depth will become $\bigO(\sqrt{|V|})$, and since the base case also runs in time exponential only in $\bigO(\sqrt{|V|} \log |V|)$, we get an overall subexponential runtime. 

As observed by Gärtner et al.\@~\cite{gartnerSubexponentialAlgorithmARRIVAL2021}, one can do better on graphs with a small feedback vertex set: The crucial ingredient is that \garrival\ can be solved efficiently on acyclic graphs by greedy simulation of the tokens. Using this as a base case and choosing pivots from a feedback vertex set yields yet another instantiation of the framework provided by the recursive algorithm. Concretely, this yields a polynomial-time algorithm on graphs with a bounded feedback vertex set.  

\subsection{Exploiting Balanced Separators}
\label{ssec:exploiting_balanced_separators}

We now describe an adaptation of 
Algorithm~\ref{algo:recursive_algorithm} that works well on graphs of small treewidth. The general idea is again to pick the pivot $p \in V \setminus T$ appropriately. Concretely, as mentioned in Section~\ref{ssec:treewidth_separators}, small treewidth ensures small balanced separators in all induced subgraphs of our input graph. Thus, it seems intuitive that in each step, we should pick the pivot $p$ from a balanced separator of the graph induced by the remaining non-terminals. Eventually, this should disconnect the induced graph into independent subinstances, each with a significantly smaller number of non-terminals. Recursing on all subinstances independently yields the desired speedup. We make this precise in Algorithm~\ref{algo:treewidth_algo} and analyse the algorithm in the remainder of this section.

\begin{algorithm}
\DontPrintSemicolon
\caption{The main difference to Algorithm~\ref{algo:recursive_algorithm} is that we choose our pivot from a smallest balanced separator $S$ that is passed through the recursive calls. As soon as all vertices in the separator have been turned into terminals, we can split the instance into independent subinstances and proceed from there. }
\label{algo:treewidth_algo}
\SetKwFunction{findSFone}{Find-Switching-Flow-1}
\SetKwFunction{findSF}{Find-Switching-Flow-2}
\Indm\findSF{$G = (V, E, s_0, s_1), T \subseteq V, (t^+_v)_{v \in T}, S \subseteq V \setminus T$}\\
\Indp
  \If{$T = V$ \tcp*{Base Case}}{
        $x(v, s_0) \gets \lceil \frac{t^+_v}{2} \rceil$ for all $v \in V = T$\\
        $x(v, s_1) \gets \lfloor \frac{t^+_v}{2} \rfloor$ for all $v \in V = T$\\
        \KwRet{$x$}
    }
    
    \If{$|S| > 0$ \tcp*{Binary Search Case}}{
        choose arbitrary $p \in S$\\
        $T' \gets T \cup \{p\}$\\
        $S' \gets S \setminus \{p\}$\\
        $\ell \gets 0$\\
        $r \gets  2^{|V|} t^+$\\
        \While{$\ell < r$}{
            $t^+_p \gets \lceil \frac{\ell + r}{2} \rceil$\\
            $x \gets$ \findSF{$G, T', (t^+_v)_{v \in T'}, S'$} \\
            \If{$x^-(p) < t^+_p$}{$r \gets t^+_p - 1$}
            \If{$x^-(p) = t^+_p$}{\KwRet{$x$}}
            \If{$x^-(p) > t^+_p$}{$\ell \gets t^+_p + 1$}
        }
    }
    \Else(\tcp*[f]{Splitting Case}){ 
        \For{every connected component $C \subseteq V \setminus T$ of $G- T$ (undirected)}{
            let $G_C$ be the graph (directed) obtained from $G$ by removing all vertices not in $C \cup T$, their outgoing edges, and replacing edges leaving the set $C \cup T$ by self-loops.\\
            find a smallest balanced separator $S_C \subseteq C$ of $G_C - T$ (undirected) \\
            $x^{(C)} \gets$ \findSF{$G_C, T, (t^+_v)_{v \in T}, S_C$}\\
        }
        combine solutions of each connected component to $x$ (see Lemma~\ref{lemma:correctness_algo_2} for details)\\
        \KwRet{$x$}
    }
\end{algorithm}

\begin{lemma}[Correctness of Algorithm~\ref{algo:treewidth_algo}]
\label{lemma:correctness_algo_2}
    Given an arbitrary \garrival-instance consisting of $G = (V, E, s_0, s_1)$ with non-empty $T \subseteq V$, starting tokens $(t^+_v)_{v \in T}$ with $t^+ \geq 1$, and a balanced separator $S$ for $G - T$, Algorithm~\ref{algo:treewidth_algo} correctly returns an integral switching flow.
\end{lemma}
\begin{proof}
    Correctness of the base case and the binary search case follows from the same arguments as in Lemma~\ref{lemma:correctness_algo_1} (correctness of Algorithm~\ref{algo:recursive_algorithm}). The only thing we changed is that our pivot is chosen from $S$ instead of all of $V \setminus T$.
    
    It remains to argue correctness of the new splitting case. For this, assume that $S = \emptyset$ and $V \neq T$. In particular, $G - T$ has at least one non-empty connected component. For each connected component $C \subseteq V \setminus T$ of $G - T$, the algorithm computes the switch graph $G_C$ obtained from $G$ by deleting all vertices $V \setminus (C \cup T)$, their outgoing edges, and replacing edges leaving $C \cup T$ by self-loops. Moreover, assume that we are given an integral switching flow $x^{(C)}$ for each of those \garrival-subinstances. Observe that every edge $e = (u, v)$ in $G$ with $u \in V \setminus T$ appears in exactly one subinstance $G_C$: Indeed, we must have $u \in C$ for some connected component $C$, and thus $e$ can only appear in $G_C$. Hence, we can uniquely assign $x(e) \coloneqq x^{(C)}(e)$ for each edge $e$ with corresponding connected component $C$. If we instead have $u \in T$, then $e$ appears in at least one subinstance $G_C$. However, since $u \in T$, the value $x^{(C)}(e)$ must be the same for all subinstances $G_C$ that $e$ appears in. Therefore, we can safely assign $x(e) \coloneqq x^{(C)}(e)$ for any of those connecteced components $C$. This fully defines $x$, and it remains to prove that it is a switching flow. This is not hard to see, again by distinguishing between vertices from $T$ and $V \setminus T$. For any vertex $u \in T$, switching behavior holds at $u$ because it holds in all subinstances (where the outgoing edges of $u$ have the exact same values even if they were converted into self-loops). Similarly, if we instead have $u \in C$ for some connected component $C$, then $x(u, s_0(u))$ and $x(u, s_1(u))$ are taken from $x^{(C)}$, where switching behavior must hold by the assumption that $x^{(C)}$ is a switching flow. Flow conservation only has to hold for non-terminals, and it holds due to the fact that all incoming edges of $u \in C$ in $G$ must be present in $G_C$ as well, implying that the flow conservation from the subinstance carries over. We conclude that $x$ is indeed an integral switching flow.
\end{proof}

\begin{lemma}
\label{lemma:splitting_soon}
    Assume that we run Algorithm~\ref{algo:treewidth_algo} on a \garrival-instance consisting of $G = (V, E, s_0, s_1)$ with non-empty $T \subseteq V$ and starting tokens $(t^+_v)_{v \in T}$ with $t^+ \geq 1$. Further assume that we input a smallest balanced separator $S$ of the subgraph $G - T$, and assume that $G$ has treewidth at most $k$. 
    Then the algorithm cannot reach recursion depth $k + 2$ without at least once recursing in a splitting case. Moreover, if a splitting case is reached, then each connected component $C \subseteq V \setminus (T \cup S)$ satisfies $|C| \leq \frac{|V \setminus T|}{2}$.
\end{lemma}
\begin{proof}
    By Lemma~\ref{lemma:existence_balanced_separator}, $S$ has size at most $k + 1$. Thus, by only using the binary search case, the algorithm can reach a recursion depth of at most $k + 1$. This implies that to reach depth $k + 2$, it must at least once have recursed in a splitting case. This happens once all vertices in $S$ have been turned into terminals, and the remaining non-terminals are $V \setminus (T \cup S)$. Since $S$ was chosen as a smallest balanced separator of $G - T$, each connected component in the graph $G - T - S$ consists of at most $\frac{|V \setminus T|}{2}$ vertices.
\end{proof}

\begin{theorem}
    Given a \garrival-instance consisting of $G = (V, E, s_0, s_1)$ with non-empty $T \subseteq V$, starting tokens $(t^+_v)_{v \in T}$ with $t^+ \geq 1$, and a smallest balanced separator $S$ of $G - T$, Algorithm~\ref{algo:treewidth_algo} computes an integral switching flow $x$ in time $2^{\bigO \left(k \log (|V \setminus T|) \log (|V | + \log t^+) \right)}$, where $k$ is the treewidth of $G$.
\end{theorem}
\begin{proof}
    As in the analysis of Lemma~\ref{lemma:correctness_algo_1}, $N  \coloneqq (1 + 2^{|V|})^{|V \setminus T|} t^+$ is an upper bound on the total number of starting tokens in any recursive call (since the recursion depth is still certainly at most $|V \setminus T|$). Compared to the analysis in Lemma~\ref{lemma:correctness_algo_1}, we now additionally have to include the splitting case, which also includes finding small balanced separators in time at most $|V|^{\bigO(k)}$. Let $\mathcal{T}(\ell, q)$ denote an upper bound on the runtime of subinstances with $\ell$ non-terminals and a set $S$ of size $q$. Observe that by Lemma~\ref{lemma:splitting_soon}, we get the upper bound
    \begin{align*}
        \mathcal{T}(|V \setminus T|, k + 1) &\leq (c \log N) \mathcal{T}(|V \setminus T|, k) \\
        &\quad \vdots \\
        &\leq (c \log N)^{k + 1} \mathcal{T}(|V \setminus T|, 0) \leq |V|^{c' k} (c \log N)^{k + 1} \mathcal{T} \left( \frac{|V \setminus T|}{2}, k + 1 \right)
    \end{align*}
    where the last inequality comes from the splitting case and accounts for the search of new balanced separators. Repeating this, we then get 
    \begin{align*}
        \mathcal{T}(|V \setminus T|, k + 1) &\leq \dots \leq |V|^{c' k} (c \log N)^{k + 1} \mathcal{T} \left( \frac{|V \setminus T|}{2}, k + 1 \right) \\ 
        &\quad \vdots \\
        &\leq \dots \leq |V|^{c' k \log{|V \setminus T|}} (c \log N)^{(k + 1) \log{|V \setminus T|} } \mathcal{T} \left( 0, 0 \right) 
    \end{align*}
    for constants $c, c'$. Using the definition of $N$ and $\mathcal{T} \left( 0, 0 \right) \leq \poly(\log N, n)$, this implies an overall upper bound of $2^{\bigO \left(k \log (|V \setminus T|) \log (|V | + \log t^+) \right)}$, as desired.
\end{proof}

\section{Reduction to \texorpdfstring{$\ell_1$}{l1}-Contraction}
\label{sec:solving_arrival_via_contraction}

The goal of this section is to prove that \garrival\ reduces to finding an approximate fixed point of an $\ell_1$-contracting function. For this, we will frequently use the functions $h_0, h_1 : \R_{\geq 0} \rightarrow \R_{\geq 0}$ defined as 
\[
    h_0(x) \coloneqq \min \left\{ x - \left\lfloor \frac{x}{2} \right\rfloor, \left\lceil \frac{x}{2} \right\rceil  \right\}  \quad \text{ and } \quad h_1(x) \coloneqq \max \left\{ \left\lfloor \frac{x}{2} \right\rfloor, x - \left\lceil \frac{x}{2} \right\rceil \right\}
\]
for all $x \in \R_{\geq 0}$. Observe that both $h_0$ and $h_1$ are continuous and monotone, and that we have $h_0(x) + h_1(x) = x$ as well as $h_1(x) \leq h_0(x) \leq h_1(x) + 1$ for all $x \in \R_{\geq 0}$.

\begin{figure}[ht]
    \centering
    \includegraphics[width=0.5\linewidth]{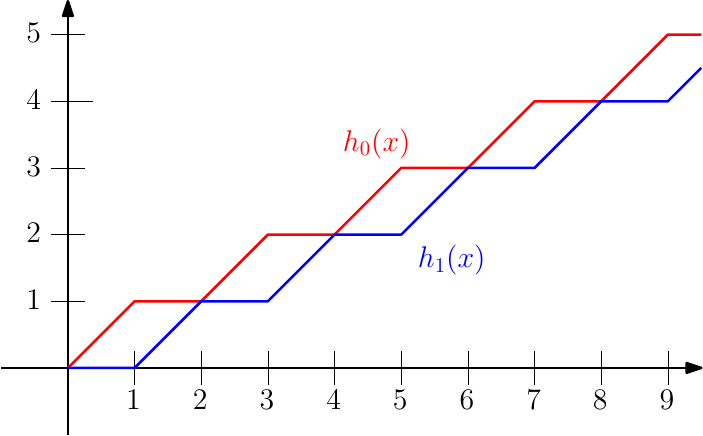}
    \caption{The functions $h_0$ and $h_1$.}
    \label{fig:enter-label}
\end{figure}

\subsection{One-Step Update}
\label{ssec:one-step_update}

Consider a vector $x \in \R^n_{\geq 0}$ and think of it as token mass that is distributed among the vertices of a switch graph, with $x_v$ token mass currently occupying vertex $v$. We want to define a notion of moving all (possibly fractional) tokens by one step each while respecting the switching rules. The following one-step update function captures this idea, with the intuition that $f(x)_v$ is the amount token mass that $v \in V$ receives from its predecessors.

\begin{definition}[One-Step Update]
    Let $G = (V, E, s_0, s_1)$ be a switch graph. The one-step update function $f : \R^{n}_{\geq 0} \rightarrow \R^{n}_{\geq 0}$ associated with this switch graph is defined as 
    \[
        f(x)_v := \sum_{u : s_0(u) = v} h_0(x_u) + \sum_{u : s_1(u) = v} h_1(x_u)
    \]
    for all $v \in V$.
\end{definition}

\begin{lemma}[Non-Expansiveness and Monotonicity]
\label{lemma:non-expansiveness_monotonicity}
    The one-step update $f : \R^n_{\geq 0} \rightarrow \R^n_{\geq 0}$ associated with the switch graph $G = (V, E, s_0, s_1)$ is monotone and non-expansive.
\end{lemma}
\begin{proof}
    Monotonicity of $f$ follows directly from monotonicity of $h_0$ and $h_1$. Thus, it remains to prove that $f$ is non-expansive. Let $x, y \in \R^n_{\geq 0}$ be arbitrary. By the previously discussed properties of $h_0$ and $h_1$, we have 
    \begin{align*}
        |x_v - y_v| &= | h_0(x_v) - h_0(y_v)| + | h_1(x_v) - h_1(y_v)  |
    \end{align*}
    for all $v \in V$. With this, we calculate
    \begin{align*}
        ||f(x) - f(y)|| &= \sum_{v \in V} | f(x)_v - f(y)_v | \\
        &= \sum_{v \in V} \left| \sum_{u: s_0(u) = v} (h_0(x_u) - h_0(y_u)) 
        + \sum_{u: s_1(u) = v}  ( h_1(x_u) - h_1(y_u) ) \right| \\
        &\leq \sum_{u \in V}  \left| h_0(x_u) - h_0(y_u) \right| + \sum_{u \in V}  \left| h_1(x_u) - h_1(y_u) \right| \\
        &= \sum_{u \in V} |x_u - y_u| \\
        &= || x - y ||,
    \end{align*}
    where we went from summing over every edge $(u, v)$ by its target $v$ to summing over every edge $(u, v)$ by its source $u$ in the inequality-step.
\end{proof}
Next, we consider what happens if we reintroduce terminals. Concretely, we want to fix $x_v = t^+_v$ for all $v \in T$ and consider the resulting function on the non-terminal vertices.

\begin{definition}[Extension and Projection]
    Let $f : \R^n_{\geq 0} \rightarrow \R^n_{\geq 0}$ be the one-step update associated with the switch graph $G = (V, E, s_0, s_1)$. Assume that we are given terminals $T \subseteq V$ with starting tokens $(t^+_v)_{v \in T}$. For arbitrary $x \in \R^{|V \setminus T|}_{\geq 0}$, let the extension $x' \in \R^n_{\geq 0}$ of $x$ denote the vector
    \[
        x'_v = \begin{cases}
            t^+_v & \text{ if } v \in T \\
            x_v & \text{ otherwise } 
        \end{cases}
    \]
    obtained by filling in the values $t^+_v$ for terminals $v \in T$. Moreover, we define the projection $g : \R^{|V \setminus T|}_{\geq 0} \rightarrow \R^{|V \setminus T|}_{\geq 0}$ of $f$ to non-terminals as $g(x)_v \coloneqq f(x')_v$
    for all $v \in V \setminus T$ and $x \in \R^{|V \setminus T|}_{\geq 0}$.
\end{definition}
Observe that the projected one-step update $g$ is still non-expansive and monotone: In particular, non-expansiveness can be obtained by
\[
    \norm{g(x) - g(y)} \leq \norm{f(x') - f(y')} \leq \norm{x' - y'} = \norm{x - y}
\]
for all $x, y \in \R_{\geq 0}^{|V \setminus T|}$.

The next lemma says that the fixed points of the projected one-step update reveal the solution to the given \garrival-instance. 

\begin{lemma}[Interpretation of One-Step Update]
\label{lemma:interpretation}
    Let $f : \R^n_{\geq 0} \rightarrow \R^n_{\geq 0}$ be the one-step update associated with the switch graph $G = (V, E, s_0, s_1)$. Assume that we are given terminals $T \subseteq V$ with starting tokens $(t^+_v)_{v \in T}$, and let $g : \R^{|V \setminus T|}_{\geq 0} \rightarrow \R^{|V \setminus T|}_{\geq 0}$ be the projection of $f$ to non-terminals. Let $x \in \R^{|V \setminus T|}_{\geq 0}$ be arbitrary and consider its extension $x' \in \R^n_{\geq 0}$. Finally, define 
    \[
        y(v, s_0(v)) := h_0(x'_v) \quad
        y(v, s_1(v)) := h_1(x'_v) 
    \]
    for all $v \in V$. Then the following two statements are true: 
    \begin{itemize}
        \item If $x$ is a fixed point of $g$ and $y(e)$ is fractional for some edge $e \in E$, then there exists a directed fractional cycle $C$ containing $e$. 
        \item $x$ is an integral fixed point of $g$ if and only if $y$ is an integral switching flow. 
    \end{itemize}
\end{lemma}
\begin{proof}
    Observe first that for every $v \in T$, both $ y(v, s_0(v))$ and $ y(v, s_1(v))$ are integral (since $t^+_v \in \N_0$ by assumption). Similarly, 
    for every $v \in V \setminus T$, either $ y(v, s_0(v))$ or $ y(v, s_1(v))$ is integral.
    
    We are now ready to prove the first statement: Assume that $x$ is a fixed point of $g$ and that $y(e_1)$ is fractional for some edge $e_1 = (u, v) \in E$. By our previous observation, we know that $x_u = y(u, s_0(u)) + y(u, s_1(u))$ must be fractional and that $u \notin T$. Since $x$ is a fixed point, this implies that $g(x)_u$ is fractional as well. By definition of $g$, this means that there must exist some edge $e_2 = (w, u) \in E$ with $y(e_2)$ fractional. We can now repeat this argument until we find a fractional cycle $C$. Observe that $C$ cannot pass through any terminals and that it must come back to $v$ and thus include $e_1$ (because at most one of the two outgoing edges at every non-terminal vertex is fractional). 
    
    For the second statement, observe that $y$ satisfies the flow conservation constraints if and only if $x$ is a fixed point of $g$: Indeed, this follows from 
    \[
        y^-(v) = \sum_{u:s_0(u) = v} h_0(x'_u) + \sum_{u:s_1(u) = v} h_1(x'_u) = f(x')_v
    \]
    and $y^+(v) = h_0(x'_v) + h_1(x'_v) = x'_v$ for all $v \in V$.
    
    Finally, integrality of $y$ immediately implies integrality of $x$ and vice versa. The definition of $y$ also implies that its values on two out-edges satisfy switching behavior. Hence, if $y$ is integral, then it must be an integral switching flow (since it automatically satisfies the switching constraints). 
\end{proof}

\subsection{Discounted One-Step Update}
\label{ssec:discounted_one-step_update}

As we have seen, the one-step update function and its projection are non-expansive with respect to the Manhattan distance. In this section, we make them contracting by artificially introducing a contraction factor. 

\begin{definition}[Discounted One-Step Update]
    Let $f : \R^n_{\geq 0} \rightarrow \R^n_{\geq 0}$ be the one-step update function and $g : \R^{|V \setminus T|}_{\geq 0} \rightarrow \R^{|V \setminus T|}_{\geq 0}$ its projection associated with the switch graph $G = (V, E, s_0, s_1)$ with terminals $T \subseteq V$ and token numbers $(t^+_v)_{v \in T}$. For $\lambda \in [0, 1)$, $f^{(\lambda)} := \lambda f$ is the $\lambda$-discounted one-step update function and $g^{(\lambda)} \coloneqq \lambda g$ its $\lambda$-discounted projection.
\end{definition}

\begin{corollary}[Contracting and Monotone Discounted One-Step Update]
    Let $\lambda \in [0, 1)$ be arbitrary. The $\lambda$-discounted one-step update function $f^{(\lambda)}$ and its $\lambda$-discounted projection $g^{(\lambda)}$ associated with the switch graph $G = (V, E, s_0, s_1)$ are monotone and contracting with parameter $\lambda$.
\end{corollary}
\begin{proof}
    Follows from Lemma~\ref{lemma:non-expansiveness_monotonicity}.
\end{proof}

\begin{lemma}
\label{lemma:unique_fixed_point_meaning}
    Let $\lambda \in [0, 1)$ be arbitrary. Assume that $g : \R^{|V \setminus T|}_{\geq 0} \rightarrow \R^{|V \setminus T|}_{\geq 0}$ is the projected one-step update and $g^{(\lambda)} : \R^{|V \setminus T|}_{\geq 0} \rightarrow \R^{|V \setminus T|}_{\geq 0}$ its discounted version associated with the switch graph $G = (V, E, s_0, s_1)$ with terminals $T \subseteq V$ and token numbers $(t^+_v)_{v \in T}$. Assume that $x^* \in \R^{|V \setminus T|}_{\geq 0}$ is the unique fixed point of $g^{(\lambda)}$. Every fixed point $x \in \R^{|V \setminus T|}_{\geq 0}$ of $g$ satisfies $x^* \leq x$. 
\end{lemma}
\begin{proof}
    Let $x \in \R^{|V \setminus T|}_{\geq 0}$ be an arbitrary fixed point of $g$. 
    We have 
    \[
        g^{(\lambda)}(x) = \lambda g(x) = \lambda x \leq x
    \]
    which implies that the function $g^{(\lambda)}$ maps the box $B := [0, x_1] \times \dots \times  [0, x_{|V \setminus T|}]$ to itself. In particular, the unique fixed point $x^*$ of $g^{(\lambda)}$ must lie inside $B$, and we get $x^* \leq x$. 
\end{proof}

\begin{lemma}
\label{lemma:sufficient_contraction_factor}    
    Let $x^\star \in \R^{|V \setminus T|}_{\geq 0}$ be the unique fixed point of the $\lambda$-discounted projected one-step update function $g^{(\lambda)} : \R^{|V \setminus T|}_{\geq 0} \rightarrow \R^{|V \setminus T|}_{\geq 0}$ associated with the switch graph $G = (V, E, s_0, s_1)$ with terminals $T$ and token numbers $(t^+_v)_{v \in T}$. Let $\lambda \in (1 - \frac{1}{t^+ + ||x^*||}, 1)$.  With 
    \[
        y(v, s_0(v)) := \lambda h_0(x^\star_v) \quad \text{ and } \quad y(v, s_1(v)) := \lambda h_1(x^\star_v) 
    \]
    for all $v \in V \setminus T$, as well as 
    \[
        y(v, s_0(v)) := \lambda h_0(t^+_v) \quad \text{ and } \quad 
        y(v, s_1(v)) := \lambda h_1(t^+_v)
    \]
    for all $v \in T$, we must have $\sum_{v \in T} y^-(v) > t^+ - 1.$
    \end{lemma}
\begin{proof}
    Let $z : E \rightarrow \R_{\geq 0}$ be an integral switching flow. Let $v \in V \setminus T$ be arbitrary and define the difference $w := z - y$. We have 
    \begin{align*}
        \underbrace{\sum_{u: (u, v) \in E} w(u, v) }_{ = w^+(v)} - \underbrace{\sum_{u: (v, u) \in E} w(v, u)}_{ = w^-(v) } &= z^+(v) - z^-(v) - y^+(v) + y^-(v) \\
        &= 0  - y(v, s_0(v)) - y(v, s_1(v))  + \sum_{u:(u, v) \in E} y(u, v) \\
        &= - \lambda x^\star_v  + g^{(\lambda)}(x^\star)_v \\
        &= (1 - \lambda) x^\star_v
    \end{align*}
    for all $v \in V \setminus T$ and 
    \[
        w^+(v) = z^+(v) - y^+(v)  = t^+_v - \lambda t^+_v  = (1 - \lambda) t^+_v
    \]
    for all $v \in T$. With $\sum_{v \in V} w^+(v) = \sum_{v \in V} w^-(v)$, we therefore get
    \[
        \sum_{v \in T}  w^-(v) = \sum_{v \in V \setminus T} (w^+(v) - w^-(v)) + \sum_{v \in T} w^+(v)  \leq (1 - \lambda) (t^+ + ||x^\star||) < 1
    \]
    by using our bound on $\lambda$.
    It remains to observe that this implies
    \[
        \sum_{v \in T} y^-(v) = \sum_{v \in T} (z^-(v) - w^-(v)) > t^+ - 1.
    \]
\end{proof}

\begin{theorem}
    Deciding an instance of \garrival\ on $n$ vertices with $t^+ \geq 1$ starting tokens reduces to finding a fixed point of the $\lambda$-discounted projected one-step update with $\lambda \in (1 - \frac{1}{t^+ (1 + n 2^{n})}, 1)$.
\end{theorem}
\begin{proof}
    By Lemma~\ref{lemma:sufficient_contraction_factor}, we know that the unique fixed point of the $\lambda$-discounted projected one-step update must send a flow of strictly more than $t^+ - 1$ to the terminals if $\lambda > 1 - \frac{1}{t^+ + ||x^\star||}$. By Lemma~\ref{lemma:interpretation} and Lemma~\ref{lemma:unique_fixed_point_meaning}, we know that this implies that any integral switching flow must send at least the same amount of flow to the respective terminals. Since every integral switching flow sends exactly $t^+$ flow to the terminals (Lemma~\ref{lemma:monotonicity_and_global_conservation}), we can infer the values $(t^+_v)_{v \in T}$ from the unique fixed point (by rounding up). The argument in Lemma~\ref{lemma:upper_bound_flow} yields the upper bound $\norm{x^\star} \leq n 2^n t^+$.
\end{proof}

Observe that the bound in Lemma~\ref{lemma:sufficient_contraction_factor} can be improved to $\sum_{v \in T} y^-(v) > t^+ - \delta$ by choosing $\lambda > 1 - \delta \frac{1}{t^+ (1 + n2^n)}$. Furthermore, any $\epsilon$-approximate fixed point 
$\hat{x}$ satisfies 
\[
    \norm{\hat{x} - x^\star} 
    \leq \norm{\hat{x} - g^{(\lambda)}(\hat{x})} + \norm{g^{(\lambda)}(\hat{x}) - x^\star} 
    \leq \epsilon + \norm{g^{(\lambda)}(\hat{x}) - g^{(\lambda)}(x^\star)} 
    \leq \epsilon + \lambda \norm{\hat{x} - x^\star}
\]
and therefore $\norm{g^{(\lambda)}(\hat{x}) - g^{(\lambda)}(x^\star)}  \leq \norm{\hat{x} - x^\star} \leq \frac{\epsilon}{1 - \lambda}$. Thus, we get 
\[
    \sum_{v \in T} \left| g(\hat{x})_v - t^-_v \right| \leq \sum_{v \in T} \left| g(\hat{x})_v - g(x^\star)_v \right| + \sum_{v \in T} \left| g(x^\star)_v - t^-_v \right| \leq \frac{\epsilon}{1 - \lambda} + \delta,
\]
which is enough to derive $(t^-_v)_{v \in T}$ from $\hat{x}$ if $\frac{\epsilon}{1 - \lambda} + \delta < \frac{1}{2}$.

Finally, capping the function $g^{(\lambda)}$ in each coordinate to make it map $[0, t^+ 2^{n}]^{|V\setminus T|}$ to itself preserves the contraction property. By Lemma~\ref{lemma:upper_bound_flow}, the unique fixed point lies inside $[0, t^+ 2^{n}]^{|V\setminus T|}$. Scaling everything by a factor of $\frac{1}{t^+ 2^{n}}$ yields a $\ell_1$-contracting function that maps $[0, 1]^{|V \setminus T|}$ to itself.

\bibliography{references}

\end{document}